\documentclass[11pt]{article}
\usepackage{times}
\usepackage[dvips]{color}
\usepackage{graphicx}
\usepackage{amsfonts,amsmath,amssymb,latexsym}
\usepackage[mathscr]{eucal}
\usepackage{algorithm,algorithmic}
\usepackage{comment}
\usepackage[utf8]{inputenc}

\def\le{\leqslant}\def\leq{\le}
\def\ge{\geqslant}\def\geq{\ge}

\newcommand{\eps}{\varepsilon}
\newcommand{\opt}{\mathrm{OPT}}

\newcommand{\cala}{{\mathcal A}}

\newtheorem{theorem}{\bf Theorem}
\newtheorem{lemma}[theorem]{\bf Lemma}

\newcommand{\sq}{\hbox{\rlap{$\sqcap$}$\sqcup$}}
\newcommand{\qed}{\hspace*{\fill}\sq}
\newenvironment{proof}{\noindent {\bf Proof.}\ }{\qed\par\vskip 4mm\par}

\begin{document}

\title{Locality-preserving allocations Problems and coloured Bin Packing}

\author{Andrew Twigg\thanks{Supported by a Junior Research Fellowship, St John's College}\\ 
  Computing Laboratory\\
  University of Oxford \\
  \texttt{andy.twigg@comlab.ox.ac.uk}  
  \and
  Eduardo C. Xavier \thanks{Supported by Fapesp and CNPq}\\
  Institute of Computing \\
  University of Campinas (UNICAMP), Brazil \\
  \texttt{eduardo@ic.unicamp.br}
}

\date{}

\maketitle
\pagestyle{headings}

%%%%%%%%%%%%%%%%%%%%%%%%%%%%%%%%%%%%%%%%%%%%%%%%%%%%%%%%%%%%%%%%%%%%%%

\begin{abstract} 
 We study the following problem, introduced by Chung et al. in 2006. We are given, online or offline, a set of coloured items of different sizes, and wish to pack
 them into bins of equal size so that we use few bins in total (at most $\alpha$ times optimal), and that the items of each 
 colour span few bins (at most $\beta$ times optimal). We call such allocations $(\alpha, \beta)$-approximate.
 As usual in bin packing problems, we allow additive constants and consider $(\alpha,\beta)$ as the asymptotic performance ratios.
  We prove
  that for $\eps>0$, if we desire small $\alpha$, no scheme can beat $(1+\eps, \Omega(1/\eps))$-approximate allocations 
  and similarly as we desire small $\beta$, no scheme can beat $(1.69103, 1+\eps)$-approximate allocations. We give offline 
  schemes that come very close to achieving these lower bounds. For the online case, we prove that no scheme can even
   achieve $(O(1),O(1))$-approximate allocations. However, a small restriction on item sizes permits a simple online scheme
    that computes $(2+\eps, 1.7)$-approximate allocations.
\end{abstract}

\section{Introduction}

We consider the problem of computing locality-preserving allocations of coloured items to bins, so as to preserve locality (colours span few bins) but remain efficient (use a few total bins). The problem appears to be a fundamental problem arising in allocating files in peer-to-peer networks, allocating related jobs to processors, allocating related items in a distributed cache, and so on. The aim is to keep the communication overhead between items of the same colour small. One application for example appears in allocating jobs in a grid computing system. Some of the jobs are related in a such a way that results computed by one job is used by another one. There are also non-related jobs that may be from different users and contexts. Related jobs are of a same colour and each job has a length (number of instructions for example). In the grid environment each computer has a number of instructions donated by its owner to be used by the grid jobs. This way the objective is to allocate jobs to machines trying to use
few machines (bins) respecting the number
of instructions available (bins size), while also trying to keep related jobs together in as few machines as possible.
In peer-to-peer systems a similar problem also appears where one want to split pieces of  files across several machines, and want to keep pieces of a file
close  together to minimize the time to retrieve the entire file.

These problems can be stated as a fundamental bi-criteria bin packing problem. 
Let $I$ be a set of items, each item of some colour $c \in C$, and denote by $I_c$ the set of items of a given colour $c$.
 Denote by $\opt(I)$ the minimum number of bins necessary to pack
all items and denote by $\opt(I_c)$ the minimum number of bins necessary to pack only items of colour $c$, i.e, as if we had
a bin packing instance with items $I_c$. Let $A(I)$ be the number of bins generated by algorithm $A$ when packing all items, and
for each colour $c$, let $A(I_c)$ be the number of bins of this packing having items of colour $c$. We say that items of colour $c$ span $A(I_c)$ bins in this packing.
We want an algorithm that minimizes both ratios $\frac{A(I)}{\opt(I)}$ and $\max_{c \in C} \frac{A(I_c)}{\opt(I_c)}$. 
So we would like to allocate the items to bins so that we use few bins in total (at most $\alpha \opt(I)$, where we call $\alpha$ the \emph{bin stretch}), and the items of each colour $c$ span few bins (at most $\beta \opt(I_c)$, where we call $\beta$ the \emph{colour stretch}). We call such allocations (or packings) $(\alpha, \beta)$-approximate. The problem of minimizing any one of $\alpha$ or $\beta$ is equivalent to the classical one-dimensional bin packing, but as we show, in general it is not even possible to minimize them simultaneously. A natural extension is to consider bins as nodes of some graph $G$, and we want to allocate bins so that each subgraph $G_c$ induced by nodes containing items of colour $c$ has some natural property allowing small communication overhead, such as having low diameter, or small size.

We prove that for $\eps>0$, if we desire small bin stretch, no scheme can beat $(1+\eps, \Omega(1/\eps))$-approximate allocations and similarly as we desire small colour stretch, no scheme can beat $(1.69103, 1+\eps)$-approximate allocations. We give offline schemes that are based in well know bin packing algorithms 
and yet come very close to achieving these lower bounds. We show how to construct $(1+\eps, \Omega(1/\eps))$ and 
$(1.7, 1+\eps)$ approximate allocations, the first one closing the gap with the lower bound and the last one almost closing the gap.
For the online case, we prove that no scheme can even achieve $(O(1),O(1))$-approximate allocations. However, a small restriction on item sizes permits a simple online scheme that computes $(2+\eps, 1.7)$-approximate allocations.

\section{Preliminaries}

We now formulate the problem of computing locality-preserving allocations
as a coloured bin packing problem. We are given a set $I$ of
$n$ coloured items  each item $e$ with a size $s(e)$ in $(0,1]$ and
with a colour $c(e)$ from $C=\{1,\ldots,m\}$, and an infinite number of
unit-capacity bins. Let $I_c$ be the set of colour-$c$ items, and
denote by $\opt(I)$ ($\opt(I_c)$ respectively) the smallest possible number of bins needed to
store items in $I$ ($I_c$ respectively). For a packing $P$ of items $I$, define $P(I)$ as
the number of bins used to pack $I$, and define $P_c(I)$ as the
number of bins spanned by colour-$c$ items in the packing $P$. When
$I$ is obvious, we drop it and write $P$ and $P_c$. 

We define an $(\alpha,\beta)$-\emph{approximate packing} as one where:
(1) $P \leq \alpha \opt(I) + O(1) $ and (2)
  for each colour $c \in C$, $P_c \leq \beta \opt(I_c) + O(1)$. An
  algorithm that always produces $(\alpha,\beta)$-approximate packings
  is called an $(\alpha,\beta)$-\emph{approximation algorithm}.

As usual in bin packing problems, we allow additive constants and consider $\alpha$ (respectively $\beta$) 
as the asymptotic performance ratio as $\opt(I)$ (respectively $\opt(I_c)$) grows to infinity (and hence the total weight of items).
This is because a simple reduction from PARTITION (eg see \cite{coffman}) shows that, without allowing additive constants, it would be NP-hard to do better than $(1.5 - \eps, \delta)$ or $(\delta, 1.5 - \eps)$ approximate packings for any $\delta$.

When dealing with the online problem we have similar definitions for
the competitive ratio of an online algorithm, and in this case $\opt(I)$ corresponds
to an optimal offline solution to instance $I$ that has full knowledge of the request sequence $I$. As standard, we shall use the term approximation
ratio interchangeably with competitive ratio when discussing online algorithms (ie a 2-approximate online scheme is one that is within a factor 2 of the optimal offline scheme).

\subsection{Related work}

Chung et al.\cite{chung06} consider the case where each item is of a different colour and can be fractionally (arbitrarily) divided between bins, bins have different sizes and the total weight of items exactly equals the total weight of bins. They show how to compute an allocation that is asymptotically optimal for each colour. By contrast, we relax the assumption that we must exactly fill all the bins, and consider the case of indivisible allocations. In this setting, the problem is much more interesting: it is impossible to get arbitrarily good $(1+\eps,1+\eps)$-approximate allocations in general. Thus, these relaxed packings have a tradeoff between bin stretch and colour stretch, with polynomial-time approximations. We also consider for the first time the case where items arrive online. However, the case of heterogenous bins is open for our setting.

The nonexpansive hashing scheme of Linial and
Sasson~\cite{linial} can also be used to find a locality-preserving
packing for unit-size items. By defining the distance of two items to
be 0 if they are of the same colour, and $\delta>1$ otherwise, one can
interpret their dynamic hashing result as follows: for any $\eps>0$,
it is possible to hash unit-size items into bins in $O(1)$ time so
that they have use $O(\opt^{1+\eps})$ bins (giving bin stretch
$O(\opt^{\eps})$ and colour stretch $O(1)$.

Krumke et al.\cite{krumke} study a related `online coloured bin packing' problem 
where the goal is to minimize the number of different colours packed
into each bin, while using the entire capacity of each bin (in their
problem all items have same unit size).
However, this problem is quite different to ours. In particular,
an optimal solution problem when minimizing the number of colours per
bin may give arbitrarily bad bin stretch.
Consider $b$ bins of capacity $x$, and unit size items of many colours
$c_1,c_2,...,c_{(x-2)b+1}$. There will be $2b$ items of colour $c_1$
and 1 item of each of the other colours. Now, a $(1,1)$-approximate packing places $x$ colours from
$\{c_2...c_{(x-2)b + 1}\}$ into each bin and the items of $c_1$ into
the remaining bins. On the other hand, a packing minimizing the
maximum number of colours per bin (while using all the capacity of each
bin) will place 2 items of $c_1$ and $x-2$ items of other other colours
into each bin. Hence, considering colour $1$, it may be packed
individually into $\opt(I_1) = 2b/x$ bins, but in this solution it
spans $b$ bins, giving colour stretch $x/2$, which can be made
arbitrarily large.

There are some other variants of bin packing problems with colours, for example the so called
Colored Bin Packing that has the restriction that items of a same colour cannot be packed 
next to each other on a same bin.  Approximation algorithms for the online version of the problem
were presented by B{\"o}hm et al \cite{bohm2014online} and Dosa and Epstein \cite{DosaE14}.
A generalization where the objective is to pack a graph $G$ into another graph $H$ where nodes into
$H$ have capacities and nodes in $G$ corresponds to items of given size, was studied by Bujt{\'a}s et al. in \cite{bujtas2011graph}

The `class-constrained bin packing problem', studied by Golubchik et al.\cite{golubchik00},
Kashyap et al.\cite{kashyap}, Xavier et al.\cite{xavier08} and
Epstein et al.\cite{EpsteinIL10} is a coloured bin packing
problem. The aim is to minimize the number of bins used, subject
to the constraint that each bin contains items from at most $c$
different colours (and subject to its capacity constraints). This problem has applications
in developing algorithms for data placement on parallel disk arrays.
Again, optimal solutions to this problem
may be arbitrarily far from good if we wish to minimize colour stretch.

A recent survey covering different variants of bin packing problems was made by Coffman et al. \cite{coffman2013bin}.

%%%%%%%%%%%%%%%%%%%%%%%%%
%
%
%%%%%%%%%%%%%%%%%%%%%%%%%%%%
\section{Impossibility results for Offline Algorithms}
\label{sec:lowerbounds}

We start by considering some lower bounds on what values of bin and colour stretch 
can be achieved simultaneously. All these bounds hold for offline
algorithms, so there is some inherent tension between the two measures
of colour stretch and bin stretch. 
We now show a lower bound on colour stretch, if we wish to take bin stretch arbitrarily small.

\begin{theorem}
For bin stretch $(1+\eps)$, it is impossible to achieve better than
$\Omega(1/\eps)$ colour stretch.
\label{thm:lb1}
\end{theorem}
\begin{proof}
For $0 < \delta < 1/2$, consider the
instance containing $n$ items of size $1-\delta$, one for each colour
$1,\ldots,n$ and $n$ items of size $\delta$, all of colour $n+1$. 
For simplicity assume that $\eps = 1/x$ for some integer $x$, and
$\delta = 2\eps$, such that $\delta n$ and $\eps n$ are integers.
%As above,
We have $\opt(I) = n, \opt(I_c) = 1$ for $c=1,\ldots,n$ and
$\opt(I_{n+1})= \delta n $. Assume we want to construct a
packing using at most $(1+\eps)n$ bins in total. 

Since items of colours $c=1,\ldots,n$ can not fit together, 
we use at most $n \eps$ bins only for colour $n+1$,
then at least $n - (n \eps / \delta) = n/2$ colour-$(n+1)$ items
overflow. Therefore at least $n/2 + \eps n = n(1/2 + \eps)$ bins
 are used for colour-$(n+1)$ items in any packing using at most
 $n(1+\eps)$ bins. Since $\opt(I_{n+1}) = \delta n$, the number of
 bins for colour-$(n+1)$ is at least 
\begin{eqnarray*}
 \opt(I_{n+1})(1/2 + \eps)/\delta & = & \opt(I_{n+1})(1/(4\eps) + 1 / 2) \\
&=& \Omega(1/\eps)\opt(I_{n+1}).
\end{eqnarray*}
The construction holds for $\delta<1/2$, so it is valid for $\eps < 1/4$.
 \end{proof}

On the other hand, if one wants to keep low colour stretch, no
bin stretch smaller than $1.69103$ can be achieved.

\begin{theorem}
For colour stretch $(1+\eps)$, it is impossible to achieve bin stretch better
than $1.69103$, for sufficiently small $\eps$.
\label{teo:lb2}
\end{theorem}
\begin{proof}
Consider the following Sylvester sequence with $l_0 = 1, l_{j+1} = l_j(l_j+1)$. For some
constant $m$, we assume we have items of $(m+1)$ different colours where,
for colour $c_i, i=0,\ldots,m$, we have a list of $n$ items each one with size $\frac{1}{l_i + 1} + \eps$, where
$\eps$ is a small enough constant that depends on the value of $m$. Note
however that $m$ is independent of $\eps$, and its value will be defined later. For each colour $c_i$ an optimal packing $\opt(I_{c_i})$ for colour stretch uses
$\frac{n}{l_i}$ bins for $i=0,\ldots,m$, each bin containing exactly $l_i$ items. We assume for simplicity
and w.l.o.g that $l_i$ divides $n$.
Let $P$ be the packing corresponding to the union of the optimal packings $\opt(I_{c_i})$ for each colour.
Notice that the bins of this packing cannot be joined together. 

Now consider an optimal packing $P^*$ for bin stretch, but which has
colour stretch at most $(1+\eps)$. We will show that if the packing has
colour stretch $(1+\eps)$, then most of the bins for each colour $c_i$
are packed like the optimal colour stretch packing. 
So $P^*$ uses almost the same number of bins as $P$.

For some colour $c_i$, let $k_j^{c_i}$ be the number of bins in $P^*$
that contain exactly $j$ items for $j=1,\ldots,l_i$.
We want to upper bound the number of bins that contain less than $l_i$
items, which is $\sum_{j=1}^{l_i-1}k_j^{c_i}$. Since $P^*$ has colour
stretch $(1+\eps)$, each colour $c_i$ must span at most $(1+\eps) n /
l_i$ bins. 
The number of bins used to pack items of colour $c_i$ can be bounded as follows:
$$k_{l_i}^{c_i} + \sum_{j=1}^{l_i - 1}k_{j}^{c_i} = \frac{n - \sum_{j=1}^{l_i-1} j k_j^{c_i}}{l_i} + \sum_{j=1}^{l_i-1} k_j^{c_i} \le (1+\eps) \frac{n}{l_i},$$
so we can write
$$ - \sum_{j=1}^{l_i-1} j k_j^{c_i} + l_i \sum_{j=1}^{l_i-1} k_j^{c_i} = \sum_{j=1}^{l_i-1}[ (l_i - j) k_j^{c_i}]  \le \eps n$$% \\
Since $(l_i -j) \ge 1$ for $j=1,\ldots,l_i-1$, 
the number of bins in $P^*$ not containing exactly $l_i$ items of colour $c_i$
is at most $\sum_{j=1}^{l_i-1}k_j^{c_i} \le \eps n.$ Hence there are at least $ n - (l_i-1) \eps n$ items that must be packed in bins
containing $l_i$ items. So at least
$$\frac{ n - (l_i-1) \eps n}{l_i} = \frac{n}{l_i} (1- (l_i-1)\eps)$$
bins are used to pack only items of colour $c_i$. Note that $m$ is a constant and then 
$(l_i-1)$ is $O(1)$.  The packing $P^*$ must use at least $$P^*(I) \ge \sum_{i=0}^m \frac{n}{l_i} (1-(l_i-1)\eps)$$
bins, while an optimal solution for bin stretch uses exactly $n$ bins
by packing one item of each colour in a bin, so the bound
\begin{eqnarray*}
\frac{P^*(I)}{\opt(I)} &\ge& \sum_{i=0}^m \frac{1-(l_i-1)\eps}{l_i} \\
&=& \sum_{i=0}^m \frac{1}{l_i} - \eps (m - \sum_{i=0}^m \frac{1}{l_i}) \ge 1.69103 - \eps'
\end{eqnarray*}
holds for $m \ge 5$ and sufficiently small $\eps \le \frac{\eps'}{ m}$.
 \end{proof}

Somewhat suprisingly, the two correct bounds are not symmetric --
the upper bounds in the next section show that we can indeed achieve $(O(1), 1+\eps)$-approximation schemes.

%%%%%%%%%%%%%%%%%%%%%%%%%%%%%%%%%%%%%%%%%%%%%
%
%
%
%%%%%%%%%%%%%%%%%%%%%%%%%%%%%%%%%%%%%%%
\section{Offline Algorithms}
\label{sec:offline}

\subsection{A $(1+\eps,O(1/\eps))$-approximation algorithm}

We now describe how to achieve asymptotically the bound in Theorem \ref{thm:lb1}. 
We shall make use of the APTAS of Fernadez de La Vega and Lueker \cite{VegaL81} (VL), which operates as follows:

\paragraph{The APTAS VL:} fix some $\eps>0$, and separate items $I$ into small $I_s$ ($<\eps$) and large $I_l$($\ge \eps$). For the large items, 
sort them by increasing size and partition them into $K=1/\eps^2$ groups, each of at most $n \eps^2$ items.
Round each item up to the size of the largest item in its group, to obtain an instance $J$.

Each bin contains at most $1/\eps$ items from $I_l$, 
so the total number of different bin types is at most $t=\binom{1/\eps + K}{1/\eps}$, and the total number 
of possible packings using at most $n$ bins is at most $\binom{n+t}{t}$, which is polynomial in $n$. Therefore we can
enumerate these packings and choose the best one. Since we have rounded all items up in size, a packing of the rounded up items
gives a valid packing of the original items. The following elegant domination argument (from \cite{VegaL81}) shows that an optimal 
packing for the rounded up items uses at most a factor $(1+\eps)$ more bins than packing the original items: consider rounding \emph{down} 
item sizes to the smallest in the group to obtain an instance $J'$. Then a packing for $J'$ gives a packing for all but the 
largest group in $J$, which contains at most $n \eps^2$ items. Since each item has size $\ge \eps$, we have $\opt(I_l) \ge n \eps$. Thus,
$$\opt(J) \le \opt(J') + n\eps^2 \le (1+\eps) \opt(I_l).$$

Now take the small items $I_s$ and pack them into the remaining free space using first fit (FF). If we do not open more bins, then we 
already have at most $(1+\eps)\opt(I)$ bins. If we need to add more bins, then clearly each bin except at most 1 is full to at least 
$1-\eps$. In this case, we have at most $\opt(I)/(1-\eps) + 1 \le (1+2\eps) \opt(I) + 1$ bins.

\paragraph{Our modification: }
For our problem, we can use the rounding step, but we cannot use the FF step for small items (as some colours may be spread over many bins).
However, a small change fixes this: group small items by colour, pack each group using FF into existing bins having more than
 $2\eps$ of free space, then open more bins if necessary. With this idea, for each colour $c$, every bin (except at most 1) either 
contains at least $\eps$ weight of colour $c$, or no items of colour $c$ (if a bin contains a large item this is
clearly true, and if not, since we used FF and each bin has at least $2\eps$ of free space, at least half of this space is used).

So each colour spans at most $\opt(I_c) / \eps$ bins, giving the desired colour stretch.
For the bin stretch, the argument is similar to the one above -- if new empty bins are used when packing small items, then 
each bin is full to at least $(1-2\eps)$ and if not, we already have the desired number of bins.

%%%%%%%%%%%%%%%%%%%%%%%%%%%%%%
%
%%%%%%%%%%%%%%%%%%%%%%%%%%%%%%%%

\subsection{A $(1.7,1+\eps)$-approximation algorithm}
\label{sec:off1.7}

We now present an algorithm that almost closes the gap with the lower
bound of Theorem \ref{teo:lb2}. For this, we will
use both the APTAS of Fernadez de La Vega and Lueker (VL) \cite{VegaL81} described above, 
and the online bin packing algorithm Bounded Best-Fit (BBF), whose competitive ratio is
1.7 \cite{CsirikJ01}.

\paragraph{BBF}: maintain at most $k$ open bins, and the rest are closed and cannot be reopened.
An item of size $s$ is packed into the open bin that is most full and has space for the item, breaking
ties arbitrarily. If no such bin exists, the fullest bin is closed and a new empty bin opened.
It is known that BBF with $k=2$ has (asymptotic) competitive ratio 1.7 \cite{CsirikJ01}.

Our algorithm is presented in Algorithm \ref{alg:1.7-approx}. It first packs items of each colour separated
using the algorithm VL. Then given all $m$ packings for each colour in some order, we apply the algorithm
BBF over the items in the order the items appears in these packings.

\begin{algorithm}[h]
\caption{$\mbox{\sc A (1.7,1+$\eps$)-approximation algorithm}$} 
\label{alg:1.7-approx}
\begin{algorithmic}[1]
\begin{small}
\STATE Arbitrarily order colours $c_1 \ldots c_m$.
\FOR{each colour $c_i$}
\STATE pack items of this colour into new bins using the APTAS of (VL) (all bins are
monochromatic).
\ENDFOR
\STATE Let $P = P(c_1) \cup \ldots \cup P(c_m)$ be all bins generated. 
\STATE Let $P'$ be a new packing initially empty.
\FOR{each item $e$ in the order it appears in $P$}
\STATE Pack $e$ into $P'$ with BBF.
\ENDFOR
\STATE Return $P'$.
\end{small}
\end{algorithmic}
\end{algorithm}
We now prove a lemma that shall be useful in proving the desired
approximation ratio of the algorithm.
\begin{lemma}
Let $b = (B_{-(k-1)},\ldots,B_{0})$,  be some opened bins that may contain items, and
let $P=(B_1,\ldots,B_x)$ be bins packing items of some set $S$. Let $P'$ be the
packing generated over the items in $S$ by the BBF algorithm in the order
they appear in $P$ using the bins in $b$ as initially opened.
Then the number of used bins by $P'$ is at most $k + x$.
\label{lemma:bbf}
\end{lemma}
\begin{proof}
Let $P' = (B_{-(k-1)}',\ldots,B_0', B_{1}', \ldots, B_{y}')$ be the bins
in the order they are closed by BBF. We will show that
any item $e \in B_i$ of $P$ for $i \in \{1,\ldots,x\}$ is packed in a bin $B_j'$ of $P'$ where $j \le i$.

Assume for contradiction that $e \in B_i$ is the first item packed in some bin $B_j'$ with $j > i$.  Since $e$ is the first
such item of $B_i$, all previous items $e' \in B_{i'}, i'=1,\ldots,i-1$ 
must have been packed in a bin $B_{j'}'$ with $j' \le i'$. So bin $B'_i$ only contains items
of $B_i$. But since $P$ is a valid packing, there must be room for $e$ in
$B'_i$.
 \end{proof}
%\medskip
\begin{theorem}
The algorithm computes $(1.7,1+\eps$)-approximate packings.% assuming at most $n$ items and $m=poly(n)$ colours.
\end{theorem}
\begin{proof}
The time bound follows since the number of colours is polynomial in
$n$, and both algorithms VL and BBF run in polynomial time.
In steps (1-4) we generate packings $P(c_i)$
for each colour $c_i$ such that $P(c_i) \le (1+2\eps)\opt(c_i) + 1$.
In steps (7-9) of the algorithm it is used the BBF algorithm to pack
the items in the order they appear in $P=(P(c_1),\ldots,P(c_m))$.
Since the BBF algorithm keeps at any time $k=2$ opened bins,
by the previous Lemma \ref{lemma:bbf}, in the final packing $P'$ we have for
each colour $c_i$,
$P'_{c_i}(I) \le (1+2\eps)\opt(I_{c_i}) + 3.$
Since BBF has approximation factor $1.7$ we also have the bound
$P'(I) \le 1.7\opt(I) + O(1)$
for the entire packing.
 \end{proof}

It is interesting to note that just packing the items of each colour in
order using bounded best fit gives a $(1.7, 1.7)$-approximate packing
(for the bin stretch, ignore colours then the entire packing is
1.7-approximate, and for colour stretch use Lemma \ref{lemma:bbf} above to get 
$P_c \le 1.7 \opt(I_c) + O(1)$ since the algorithm is bounded space).

Note that in order to have the $(1.7,1+\eps)$-approximation we need a bounded
space online algorithm on step 8 of Algorithm \ref{alg:1.7-approx}, but not necessarily
any online bounded space algorithm would work. We have to use an algorithm that satisfies
the property of Lemma \ref{lemma:bbf}. Consider the Harmonic algorithm \cite{LeeL85} for instance.
If we have an instance consisting of just one colour, then after step 8, the Harmonic algorithm would
have separated items by types of sizes creating an entire new packing and the $(1+\eps)$-colour stretch
would be lost.

%
%
%
%
%
%
%
%
%
%\medskip
%
%
%
%
%
%
%
%%%%%%%%%%%%%%%%%%%%%%%%%%%%%%%%%%%%%%%%%%%%%%%%%%%%%%%%%%%%%%%
%
%
%%%%%%%%%%%%%%%%%%%%%%%%%%%%%%%%%%%%%%%%%%%%%%%%%%%%%%%%%%%%%
\section{Online Algorithms}%for \\ non-divisible items}
\label{sec:nondivisible}

We now consider the online version of the problem. Coloured items
arrive and must be packed with no knowledge of future arrivals.
The main difficulty with constructing an online algorithm is that we
don't know in advance the total weight of each colour, but on the other hand would like
to reserve space so that colours of small weight aren't spread over many bins.

\subsection{Impossibility of online $(O(1),O(1))$-approximation}

Even under the restriction that items are $\ge \eps$, there is
still a lower bound $L \geq 1.5403$ for the online classical bin packing
problem, due to by Balogh et al. \cite{Balogh12} which improved  
a previous lower bound of $1.5401$ by van Vliet \cite{144239}.
 In this case, 
$L$ is also a lower bound for both bin stretch and colour stretch; we cannot
hope to do better in either parameter.
To see this, consider packing items of only one colour, then 
the number of bins used cannot be smaller than $L \opt + O(1)$.

We now show that no online $(O(1),O(1))$-approximation scheme 
exists. The idea is to consider items in rounds. In each
round an optimal packing for the items needs a single extra bin.
If a scheme has bin stretch $O(1)$ it only needs $O(1)$ bins per round,
but if the instance has a large number of colours then some fraction of colours
will be forced to split every few rounds, from which it follows that
some colour must split in at least a constant fraction of rounds.

\begin{theorem}
\label{thm:online-lowerbound}
There is no $(\alpha,\beta)$-approximation online algorithm for the coloured bin
packing problem where $(\alpha,\beta)$ are constants.
\end{theorem}
\begin{proof}
Consider an instance where there are $n$ different colours and where each item
has size $1/n$. We analyze the packing in rounds and in each round we receive a list $L$ of $n$
items of $n$ different colours. There will be at most $nx$ rounds where
$x$ is the number of bins necessary to pack any colour after $nx$ rounds, and
$n$ is going to be defined later.

Let $I_i$ be the total number of items until round $i$  (there are $ni$) . Since
the optimal packing for bin stretch uses $i$ bins, any
$(\alpha,\beta)$-algorithm must use at most $\alpha i$ bins by round $i$.
Wlog assume that in each new round 
the algorithm uses at most $\alpha$ bins, since otherwise it will have approximation ratio $>\alpha$
and the adversary will stop at this point. The algorithm may use less than $\alpha$ bins
in one round and more than $\alpha$ bins in a later round, but the average per round must
be $\alpha$. So we can assume that the algorithm opens $\alpha$ bins per round even if
it will only use some bins in a later round. So we focus on the colour stretch.
We will consider at most
$nx$ rounds, and so there are at most $nx$ items of each colour.
Clearly the optimal packing for colour stretch has $x$ bins per colour, so
$\cala$ must guarantee that each colour spans at most $x \beta$ bins.

In each round the algorithm has $\alpha$ bins available and some other bins
that were partially filled. Since the bin stretch is guaranteed to be at most
$\alpha$, the only job of the algorithm is to keep the $\beta$
approximation in colour stretch. We now show that any algorithm
$\cala$ must incur colour stretch larger than $\beta$ on the request
sequence.

Define $s(i)=\alpha+\alpha^2+\ldots+\alpha^i$, with $s(0)=0$. For $i \ge 1$, stage $i$ consists of rounds $(s(i-1)+1)\ldots s(i)$. At the end of stage $i$ the algorithm has $\alpha s(i)$ bins available, and items of total weight $s(i)$. It can be seen that
\begin{equation}
\label{eqn:si+d}
s(i+d) = \alpha^d s(i) + s(d) >
\alpha^d s(i).
\end{equation}
We consider groups of $d+1$ stages, where $d = \lceil \frac{5}{\log \alpha} \rceil$. 
We will show that for every group of $d+1$ stages, there exists a set of at least $n/8$ colours that split during this group, i.e colours
that need to be packed in more bins than the ones available in the beginning of the group.

Here is the proof of this claim. Assume that less than $n/8$ colours split during the first $d$ stages of the group starting at
round $s(i)+1$. Then there is a set of at least $7n/8$ colours that do not split during the next $d$ stages. All items
of these colours remain packed in the first $\alpha s(i)$ bins. In this case,  we have items of weight 
\begin{eqnarray*}
\frac{7}{8} (s(i+d) - s(i)) + s(i) &=& \frac{1}{8} (7s(i+d) + s(i)) \\
&>& \frac{7}{8} \alpha^d s(i)
\end{eqnarray*}
going into at most $\alpha s(i)$ bins (the last inequality uses (\ref{eqn:si+d})).
Now we choose $d$ so that $\frac{7}{8} \alpha^d s(i) > \alpha s(i) $,
which is satisfied by taking 
$d= \lceil \frac{5}{\log \alpha} \rceil >\frac{\log 8/7}{\log \alpha} + 1.$

This shows that at
least $7n/8$ ($>n/8$) colours (say $C'$) must split during the following (i.e. $(d+1)$th)
stage. The argument for this is the following: consider all the bins
that contained items of colours $C'$ at the start of the group. All these
bins become overfull just by considering the weight of items with
colours in $C'$. So for every colour $c$ in $C'$, at least one of its bins
splits, and so at least $|C'|$ colours split.

Clearly, an item of every colour is contained in some bin at the start
of every group, so the claim implies that after $q(d+1)$ stages, we
have at least $qn/8$ splittings. So taking $q > 8x\beta$, after 
$ \frac{125 x \beta}{\log \alpha} > 8x \beta (d+1)$
 stages, we have had $>x \beta n$ splittings, so some colour must have split $> x \beta$ times.
It remains to choose $n$ large
enough so that we have at least $\frac{\alpha}{\alpha-1} 2^{125 x \beta} > s(\frac{125 x \beta}{\log \alpha})$
rounds.
 \end{proof}

%%%%%%%%%%%%%%%%%%%%%%%%%%%%%%%%%%%%%%%%
%
%%%%%%%%%%%%%%%%%%%%%%%%%%%%%%%%%%%%%%%%
\subsection{An online $(3,1.7)$-approximation}
\label{sec:online1}

In this section we provide an online algorithm that computes
$(3,1.7)$-approximate packings, but we need to assume that each item has size at least $\eps>0$, where $\eps$ is a constant. 
Note that in the approximation factors there is a dependency on the value of $\eps$, since the bin  stretch is limited by $3\opt(I) + O(\log 1/\eps)$
and the colour stretch is limited by $1.7\opt(I) + O(\log 1/\eps)$.

Wlog assume that $\eps = \frac{1}{2^j}$ for some positive integer
$j$. We consider two types of bins: isolated, that corresponds to bins
packing only items of a given colour, and non-isolated, that may pack items
of different colours. For each $i=1, \ldots, j$ we define some special bins which we call level-$i$ bins. 
A level-$i$ bin is divided in exactly $1/2^i\eps$ regions each of size $ 2^i \eps$ for $i=1,\ldots,j$.
These regions are monochromatic (each region contains items
of at most one colour). A region in some level-$i$ bin is called a \emph{level-$i$ region}.
We use a modified NF algorithm MNF to pack items into non-isolated
bins, and switch to BBF to pack colours in isolated bins. MNF is
similar to NF and a description is given below (Algorithm
\ref{alg:MNF}).

\begin{algorithm}[h]
\caption{$\mbox{\sc Modified Next Fit (MNF)}$} \label{alg:MNF}
\begin{algorithmic}[1]
\begin{small}
\STATE To pack item $e$ of colour $c$ and size $s(e)$
\STATE{Let $i$ be the highest level where an item of colour $c$ is packed}
\STATE{(Let $i=1$ if this is the first item of colour $c$)}
\IF{$e$ can be packed in the level-$i$ region}
\STATE Pack $e$ into this level-$i$ region
\ELSE
\STATE Let $l > i$ be the lowest level such that $2^l \eps \ge s(e)$
\STATE Pack $e$ into a new level-$l$ region (possibly creating a new level-$l$ bin)
\ENDIF
\end{small}
\end{algorithmic}
\end{algorithm}

Notice that each colour occupies at most one level-$i$ bin, for each
level $i=1,\ldots,j$.
We say that a colour $c$ has level $i$ if $i$ is the largest level of a bin
containing items of colour $c$. 
The following algorithm uses MNF to pack items of the same colour
until the colour has level $j$. When 
this happens the algorithm starts packing items of the colour in isolated
bins (the last level-$j$ bin is also considered an isolated bin) using
the BBF algorithm. A description of the algorithm is given below (Algorithm \ref{alg:online_norepacking}).

\begin{algorithm}[h]
\caption{$\mbox{\sc A (3,1.7)-approximation
    algorithm.}$}  \label{alg:online_norepacking}
\begin{algorithmic}[1]
\begin{small}
\STATE To pack item $e$ of colour $c$ and size $s(e)$
\IF{colour $c$ has level $<j$}
\STATE Pack $e$ with MNF in the non-isolated bins
\ELSE
\STATE Pack $e$ with BBF in the isolated bins.
\ENDIF
\end{small}
\end{algorithmic}
\end{algorithm}

A region is used when there are items packed on it. A bin is used when
all its regions are used.
The following lemma states that level-$i$ bins that are used, have at least
$1/3$ of their capacity used by items. Notice that for each level $i$, at most one
level-$i$ bin is not using all its regions, since a new level-$i$ bin
is created only when all existing level-$i$ regions are used. So 
there are at most $O(\log \frac{1}{\eps})=O(1)$ non-isolated bins that have
some unused regions.
\begin{lemma}
\label{lem:pairlemma}
Consider the non-isolated bins that have all their regions in use.
On average, each bin has at least $1/3$ of its capacity used by items.
\end{lemma}
\begin{proof}
We will prove this by considering the levels used by any colour $c$
using non-isolated bins.

For each colour $c$, a \emph{group} is a maximal sequence of regions
$2^k\eps, 2^{k+1}\eps, \ldots,2^{k+p}\eps$  used by colour $c$
(each colour may occupy a number of disjoint groups). We will show that
for each group, its regions used by colour $c$ have $1/3$ of their area
occupied. Let $2^k\eps,2^{k+1}\eps, \ldots,2^{k+p}\eps$ be a group
used by colour $c$.

We have two cases:
\begin{itemize}
\item {\bf p is odd:} Consider the pairs of adjacent regions
$$ (2^k\eps, 2^{k+1}\eps), \ldots,(2^{k+p-1}\eps, 2^{k+p}\eps).$$
Since we used MNF to pack the items, for each pair of regions the
total weight of items is at least the size of the
region in the lowest level. 
Since the higher level region is twice the
size of the lower one, each pair has at least $1/3$ of its area occupied.

\item  {\bf p is even:} If $k\ge 2$ then there is an item in the first
region $2^k\eps$ of the group that could not fit in a previous used
region by colour $c$. This item was packed in the smallest region with
room for it. So this item occupies at least $1/2$ of region $2^k
\eps$. If $k = 1$ then the assumption that $s(e) \ge
\eps$ implies that this region is filled by at
least $1/2$. The remaining regions $2^{k+1}\eps,
\ldots,2^{k+p-1}\eps, 2^{k+p}\eps$ can be paired as in the odd case,
and for each pair at least $1/3$ of its total area is occupied.
\end{itemize}
 \end{proof}

With this result we can prove the following theorem:
\begin{theorem}
The algorithm is a $(3, 1.7)$-approximation scheme, and uses space
at most $O(m)$.
\end{theorem}
\begin{proof}
Since we use BBF to pack isolated bins we can guarantee that
on average at least $1/2$ of the area of the isolated bins is occupied,
and for the non-isolated bins, the previous lemma says that at least
$1/3$ of the capacity of these bins is occupied, with the exception of
at most $O(\log 1/\eps)$ bins. So for bin stretch we have a bound of
$3\opt(I) + O(\log 1/\eps)$.

Now we consider the colour stretch. For a colour $c$ using only
non-isolated bins, it must use at most $O(\log 1/\eps)$ bins, 
which is a constant. If a colour $c$ also uses isolated bins, then by the performance bound
of BBF \cite{CsirikJ01}, it uses at most $1.7\opt(I_c) + O(\log 1/\eps)$
bins. The approximation ratio $(3,1.7)$ is then valid if $1/\eps$ is bounded by a
constant.

For the space bound, BBF uses at most $O(1)$ open bins per isolated
colour, and MNF uses at most one open bin per level which is $O(\log 1/\eps)$, a constante.
 \end{proof}

We may also consider trying to improve the bin stretch bound of 3 by
using a variation of FF instead of NF in the modified next fit
scheme. A `modified first fit' MFF works as MNF except that step 2
is replaced by `let $i$ be the lowest level occupied by colour $c$ with space for $e$', 
so that we try to pack $e$ in each of the used regions, packing it
in the first such region with space for $e$. This
requires at most $O(\log 1/\eps)$ open bins per colour, but one may
expect better performance, bearing in mind that FF beats NF.
The following result shows that this is not the case,
and using MFF provides no improvement in the bin stretch.

\begin{theorem}
Using either MFF or MNF, the approximation factor $(3,1.7)$ is tight.
\end{theorem}
\begin{proof}
Consider that $\eps = \frac{1}{2^j}$ for some positive even integer
$j$. 
Notice that the regions size are $(2\eps, 4\eps, \ldots , 1/2, 1)$ which
is equal to $(\frac{1}{2^{j-1}}, \frac{1}{2^{j-2}}, \ldots, \frac{1}{2}, 1)$.
We will consider pairs of colours $(c,c')$. Assume that for colour $c$ we receive $j$ items
in the following order: an item of size $\frac{1}{2^{j-i}}$ followed
by an item of size $\frac{1}{2^{j-i}}+\gamma$, for
$i=0,2,6,\ldots,j-2$, where $\gamma >0$ is arbitrarily small. Then we
receive items of colour $c'$. For
colour $c'$ we have $j-2$ items in the following order: an item of size
$\frac{1}{2^{j-i}}+\gamma$ followed by an item of size
$\frac{1}{2^{j-i}}$, for $i=1,3,\ldots,j-3$.

Using the MFF (or MNF) algorithm to pack these items, for colour
$c$ we will have items $\frac{1}{2^j},\frac{1}{2^j} +\gamma, \ldots, \frac{1}{2^2}, \frac{1}{2^2}+\gamma$
packed respectively in regions $\frac{1}{2^{j-1}}, \frac{1}{2^{j-2}}, \ldots, \frac{1}{2}, 1$.
For colour $c'$ we will have items $\frac{1}{2^{j-1}}+\gamma,\frac{1}{2^{j-1}}, \ldots, \frac{1}{2^3}+\gamma, \frac{1}{2^3}$
packed respectively in regions $\frac{1}{2^{j-2}}, \frac{1}{2^{j-3}}, \ldots, \frac{1}{2^2}, \frac{1}{2}$.
Pairing items of each colour we can see that they use approximately
1/3 of the allocated area for each pair (for small $\gamma$). 
So for each colour it uses approximately $1/3$ of the total area allocated to it.
An optimal packing of the items of the colours $(c,c')$ uses one bin almost
full. To see this, note that the sum of the sizes of the items of these colours is
$$\sum_{i=2}^j \frac{1}{2^i}  +  \sum_{i=2}^j (\frac{1}{2^i}  +\gamma) = \sum_{i=1}^{j-1}\frac{1}{2^i}+ (j-1) \gamma \le \sum_{i=1}^{\infty} (1/2)^i =1$$
for sufficiently small $\gamma$ (i.e. $\frac{1}{(j-1)\cdot 2^j}$). So for
appropriate values of $j$ and letting $\gamma \rightarrow 0$, the sum of item sizes 
of each pair of colours can be made arbitrarily close to 1.
We can then consider arbitrarily large instances by using many pairs of colours, thus establishing
an asymptotic lower bound on 3 on bin stretch.

In the same instance we consider a special colour $c^*$ where we first
receive an item of size 1 and then an instance that provides the worst
case ratio $1.7$ for the BBF algorithm (see \cite{JohnsonDUGG74}).
Then for colour stretch the bound $1.7$ is also tight.
 \end{proof}
%

%%%%%%%%%%%%%%%%%%%%%%%%%%%%%%%%%%%%
%
%
%%%%%%%%%%%%%%%%%%%%%%%%%%%%%%%%%%%%%%%
\subsection{An online $(2+\eps,1.7)$-approximation}
\label{sec:online2}

In this section we show how to extend the algorithm
of the previous section \ref{sec:online1} to get an online algorithm that computes
$(2+\eps,1.7)$-approximate packings. We assume that each item 
has size at least $\eps>0$.

The algorithm also uses isolated and non-isolated bins. It pack items of colour $c$ in
non-isolated bins while the total size of packed items of colour $c$, $w(c)\le g$, where
$g=1/\eps$. When $w(c)>g$ the algorithm uses isolated bins to pack items of colour $c$.
The algorithm is given below (Algorithm \ref{alg:online3}).

\begin{algorithm}[h]
\caption{$(2+\eps,1.7)\mbox{\sc-approximation}$} \label{alg:online3}
\begin{algorithmic}[1]
\begin{small}
\STATE Let $g = 1/ \eps$
\STATE For each colour $c$ let $w(c)=0$
\STATE To pack item $e$ of colour $c$ and size $s(e)$
\IF{ $w(c) \le  g$}
\STATE Pack $e$ into non-isolated bins using FF
\STATE $w(c) \leftarrow w(c) +s(e)$
\ELSE
\STATE Pack $e$ into isolated bins of colour $c$ using FF
\ENDIF
\end{small}
\end{algorithmic}
\end{algorithm}

\begin{theorem}
The algorithm is a $(2+\eps,1.7)$-approximation.
\end{theorem}
\begin{proof}
First consider colour stretch. For each colour $c$, it uses at most $g/ \eps$
non-isolated bins, because each item has size at least $\eps$. When $w(c) > g$
it packs all items of this colour in isolated bins, and since we use the FF algorithm we
can bound colour stretch by $1.7 \opt + O(g/ \eps)$.

Now for bin stretch we have the following. At the end of the execution of the algorithm,
it uses $N_1$ non-isolated bins. It uses some isolated bins as well for large colours 
(the ones with $w(c)>g$). There are some large colours that uses 
just one isolated bin, and assume there are $k$ of these large colours. There are
some other large colours that uses more than one isolated bin, and assume that
in total the algorithm uses $N_2$ bins for these large colours. 

Since we used FF to pack the items in non-isolated bins we know that
on average each one of the $N_1$ bins are full to 1/2. For the same reason
the $N_2$ bins used by large colours that uses more than one bin are full
on average by at least 1/2. So the following bound is valid 
$\opt \ge \sum_e s(e) \ge (N_1 + N_2)/2.$

Notice that the algorithm starts to use isolated bins for colour $c$, only when $w(c)>g$, and
then we also have the bound
$\opt \ge \sum_e s(e) \ge  k\cdot g.$

The approximation ratio of the algorithm is then bounded as follows
\begin{eqnarray*}
\mbox{ratio} &= &\frac{N_1 + N_2 + k}{\opt} \\
                    &\le &\frac{N_1 + N_2 +k}{\sum_e s(e)}\\
                    & \le & \frac{N_1+N_2}{(N_1+N_2)/2} + \frac{k}{gk}\\
                    & = & 2 + 1/g.\\
\end{eqnarray*}
Since $g=1/ \eps$ the algorithm is a $(2+\eps,1.7)$-approximation.
 \end{proof}

%%%%%%%%%%%%%%%%%%%%%%%%%%%%%%%%%%
%
%
%
%%%%%%%%%%%%%%%%%%%%%%%%%%%%%%%%%%%
\section{An APTAS to approximate optimal bin stretch}
\label{sec:ptas}

Our lower bounds show that it is impossible in general to achieve a $(1+\eps,1+\eps)$-approximation.
We now show that for any colour stretch $\beta > 1$ 
and every class of instances that admit a $(1,\beta)$-approximate packing,
 we can compute a $(1+\eps,(1+\eps)\beta)$-approximate packing in polynomial time.
In this section we consider that the number of different colours $m$ is bounded by a constant; relaxing this restriction remains open.
It is also worth noting that our problem now is slightly different since we are assuming instances that admit a $(1,\beta)$-approximate
packing, and in general there are instances that do not admit such packings (see Theorem \ref{thm:lb1}). 

We shall use $\opt_{\beta}(I)$ to denote the smallest number of bins needed to pack items $I$ using
 colour stretch $\le \beta$. Similarly (by slight abuse of notation), $\opt_{\beta,c}(I)$ is the number of bins spanned by
colour $c$ in such a packing. Clearly, $\opt_{\beta,c}(I) \leq \beta \opt(I_c)$. 
The scheme we describe below computes a packing $P$ satisfying \emph{(1)} total number of bins $P(I) \leq (1+\eps) \opt_{\beta}(I) + O(1)$, and \emph{(2)} for each colour $c$, it uses at most $P_c(I) \le (1+\eps) \opt_{\beta,c}(I) + O(1) \le \beta (1+\eps) \opt(I_c) + O(1)$ bins.

The idea is to use a variant of the grouping and rounding technique, but to explicitly work with the instance where items are rounded down. We are able to show that by packing some very large items and very small items separately, only a few items `overflow' from some optimal packing OPT, and thus we can still achieve the desired colour stretch and bin stretch.

Denote by $I^l$ the items in $I$ with size at least $\eps^2$ (large items), and $I^s$ the remaining items in $I$ (small items).

\paragraph{Packing large items:}
Partition the large items by colour: $I^l = I_1,\ldots,I_m$
and let $n_c =|I_c|$ be the number of items of each colour $c$. Then 
sort each colour $I_c$ by decreasing order of item size and 
partition it into at most $M = \lceil 1/ \eps^3 \rceil$ groups
$I_{c1},I_{c2},\ldots,I_{cM}$, i.e $I_c = I_{c1} \| \ldots \| I_{cM} $ where $\|$ is a concatenation operator.
Each group has $\lfloor n_c \eps^3 \rfloor$ items except perhaps the last.

For each group of each colour, round down the items to the size of the smallest item in the group
(by contrast with VL, who round \emph{up} item sizes).
As before, we can enumerate all such packings: the number of items per bin is at most $y \le 1/\eps^2$, and the
number of distinct item sizes is a constant $mM$ (recall $m$ is the number of colours and is assumed to be a constant).
Thus, there are at most  $r'=\binom{y+Mm}{y}$ different bin configurations.
We shall do something more involved with the small items, so we shall attach to each 
bin configuration a subset of colours that shall be used for the small items later on.
This gives at most $r=r' 2^m$ total configurations so the number of feasible packings
into at most $n$ bins is bounded by $\binom{n+r}{n} \le (n+r)^r$.
Notice that among the configurations there are some that may contain bins with no large items,
and just have a subset of colours attached to show that small items can be packed later.

We enumerate all such packings, and keep only those that have colour stretch at most $\beta$
(ignoring the additive constant). One of these packings corresponds exactly to
an optimal packing after removing its small items and with its large items rounded down.

A similar domination argument to before will now show that at least one of these packings has close to 
the desired colour stretch and bin stretch. 
Let $P$ be one of the enumerated packings with colour stretch at most $\beta$.
Since item sizes were rounded down, each group $I_{cj}$ in $P$ gives a packing for 
the items with orignal sizes in the next group $I_{c(j+1)}$ (all these items have smaller size than the previous group). 
The only items not packed by this are those in the first group (with largest size) -- denote these items by $Q = \cup_{c=1}^m I_{c1}$.

The `very large' items in $Q$ are packed into new bins using first fit (FF), considering all items of one colour
before the next colour. Let $P(Q)$ be the size of the packing obtained in this way.
The following simple argument shows that these very large items will contribute
only a small amount to the total bin and colour stretch.
\begin{lemma}
\label{lema32}
$P(Q) \le  \eps \opt(I)$ and $P(Q_c) \le  \eps \opt(I_c)$ for each colour $c$.
\end{lemma}
\begin{proof}
Clearly first fit packs at least one item per bin. Since $|Q| \le \sum_c n_c \eps^3$,
and each item has size at least $\eps^2$, we have
$P(Q) \le  \eps \opt(I)$. Since we apply first fit to items
grouped by colour, the same argument establishes the claim for colour stretch.
 \end{proof}

\paragraph{Packing small items.}
Let $P=\{B_1,\ldots,B_k\}$ be a packing of the large items $I^l$.
We now wish to pack carefully the small items $I^s$ into $P$.

The packing of the small items is obtained from a
solution of a linear program. Recall that when enumerating packings of large items, each bin
was tagged with a subset of colours that could be used to pack small items.
Let $N_i \subseteq \{1,\ldots,m\}$ be the set of possible
colours that may be used to pack the small items in the bin $B_i$ of the
packing $P$. For each colour $c \in N_i$, define a non-negative variable
$x_c^{i}$. The variable $x_c^{i}$ indicates the total size of small items of
colour $c$ to be packed in the bin $B_i$. Denote by $s(B_i)$ the total size of
items already packed in the bin $B_i$. Consider the program denoted by LPS:
\begin{equation*}
     \begin{array}{rllr}
          \max \displaystyle \sum_{i=1}^k \sum_{c \in N_i}  x_c^{i} &  \mathrm{s.t.} & \\
             \displaystyle  s(B_i) + \sum_{c \in N_i} x_c^i  & \le 1 &  i =1,\ldots, k & (1) \\
             \displaystyle  \sum_{i=1}^k  x_c^{i} & \le s(I^s_c) & c =1,\ldots,m & (2)\\
     \end{array}
\end{equation*}
where $I^s_c$ is the set of small items of colour $c$ in $I$.
The constraint (1) guarantees that the total size of items packed in each bin
does not exceed the bins size and constraint (2) guarantees that the sum of the
values of variables $x_c^{i}$ is not greater than the total size of small
items.

Given a packing $P$, and a fractional packing of the small items, we 
do the following: for each variable $x_c^{i}$ we pack,
while possible, the small items of colour $c$ into the bin
$B_i$, so that the total size of the packed small items is at
most $x_c^{i}$.  The possible remaining small items (the `overflowing' items) are 
packed using FF into new bins, again grouped by colour (meaning pack all items of one colour before the next).

\begin{algorithm}[h]
\caption{$\mbox{\sc APTAS$(I)$}$} 
\label{alg:PTAS}
\begin{algorithmic}[1]
\REQUIRE Number of different colours $m$ in $I$ is $O(1)$
\STATE Fix $\eps > 0$
\STATE Split items $I$ into small ($< \eps^2$) $I^s$ and large ($\ge \eps^2$) $I^l$
\STATE Group large items by colour and sort by decreasing size
\STATE Group large items of each colour into $\lceil 1/\eps^3 \rceil$ groups and round item sizes down in each group
\STATE Enumerate all packings of large items, with attached `small colours' labels
\STATE For each colour $c$, pack remaining `very large' items $Q_c$ using FF
\\
\STATE For each packing $P$, solve LPS to add small items
\STATE For each packing $P$, pack the overflowing small items into new bins using FF
\STATE Return the best packing that has colour stretch at most $\beta$
\end{algorithmic}
\end{algorithm}

\paragraph{Approximation ratio.}
We will claim that there exists a packing $P$ such that after the very large items $Q$
and the small items $I^s$ have been packed into $P$, it has the desired bin and colour stretch.
In particular, at least one packing uses at most $(1+O(\eps))\opt_{\beta}(I) + O(1)$ bins in total and
at most $(1+O(\eps))\opt_{\beta,c}(I) + O(1) \le \beta(1+O(\eps))\opt(I_c) + O(1)$
bins for each colour $c$.

\begin{theorem}
Let $\beta>1$ be the desired colour stretch.
The algorithm finds a packing $P$ such that
$P(I) \le (1+O(\eps))\opt_{\beta}(I) + O(1)$, and  
$P_c(I) \le \beta (1+O(\eps))\opt(I_c) + O(1)$ for each colour $c$.
\end{theorem}
\begin{proof}
Let OPT$_{\beta}$ be an optimal packing for the instance $I$ with colour stretch $\beta$.
Let OPT$'_{\beta}$ be the packing OPT$_{\beta}$ without the small items and with the large items rounded down as described.
Assume that each bin of OPT$'_{\beta}$ has an indication of the colours of small items used in the corresponding bin of
OPT$_{\beta}$. Clearly in the enumeration step of the algorithm one  packing with the same configuration of OPT$'_{\beta}$ with 
rounded items, is generated.  This gives a packing $P$ for the original items in $I^l \setminus Q$. 
Notice that the number of bins used by $P$ and OPT$_{\beta}$ is the same. The very large items in $Q$ are packed  separately.

In the packing $P$ there must be enough room to pack all small items, since there is in OPT$_\beta$. When packing
the small items (guided by the fractional packing LPS), at most one small item of each colour is not packed into each desired bin.
So, the total size of small items that overflow and need to be packed into new bins is at most OPT$_{\beta}(I)\eps^2m$.
These small items use at most
$ \left\lceil \frac{ \mbox{OPT}_{\beta}(I)\eps^2m}{(1-\eps^2)} \right\rceil  +  1 $
new bins, since each bin is full to at least $(1-\eps^2)$ except
perhaps by the last one.
Considering colour stretch, each colour $c$ uses at most
$\left\lceil \frac{ \mbox{OPT}_{\beta}(I_c)\eps^2}{(1-\eps^2)} \right\rceil +  1 $
new bins.

The algorithm packs these small items in new bins obtaining a new
packing $P'(I \setminus Q)$. The number of bins is at most
\begin{eqnarray}
      P'(I \setminus Q)  &\leq & \opt_{\beta}(I) + \left\lceil  \frac{\opt_{\beta}(I)\eps^2m}{1- \eps^2} \right\rceil +1 \\
                  & \leq & (1 + O(\eps))\opt_{\beta}(I) + O(1).
\end{eqnarray}

Considering colour stretch we have, for each colour $c$,
\begin{eqnarray}
      P'(I_c \setminus Q_c)  &\leq & \opt_{\beta}(I_c) + \left\lceil  \frac{\opt_{\beta}(I_c)\eps^2}{1- \eps^2} \right\rceil +1 \\
                  & \leq & (1 + O(\eps))\opt_{\beta}(I_c) + O(1).
\end{eqnarray}

To finish the proof, it remains to consider the very large items $Q$. For these, Lemma \ref{lema32} shows that they need an extra $\eps$ fraction of bins for each colour and in total.  Notice that in order to obtain a truly $(1+\eps)$ approximated solution, we need to rescale the value of $\eps$, for example
by using $\eps' = \eps/m$ due to the factor $m$ multiplying $\eps$ in the term $\mbox{OPT}_{\beta}(I)\eps^2m$ on equation (2).
So the running time of the entire algorithm is dominated by the enumeration step which is $ O(n^{O(1/\eps'^2)^{m/\eps'^3}}) =  O(n^{O(m^2/\eps^2)^{m^4/\eps^3}})$
 \end{proof}

\section{Open Problems}
\label{sec:conclusion}

\begin{itemize}
\item {\bf Improved approximation ratio.}
  Can we get an online algorithm with $(1.7+\eps,1.7+\eps)$
  approximation ratio, with the assumption on minimum item sizes?

\item {\bf Multicoloured items.} 
  The multicoloured case is also interesting: fix a set of
  (possibly unbounded) colours $\mathcal{C}$, and let each item have
  several (say at most $k$) colours from $\mathcal{C}$. The original
  definitions of colour and bin stretch still apply. By allowing $k$
  copies of each item to be packed, it is certainly possible to
  reuse any $(\alpha,\beta)$-approximation algorithm in this paper to
  construct one with bin stretch $k\alpha$ and colour stretch
  $\beta$. Is it possible to do better?

\item {\bf Network packing version.}
  Let the items form a graph as follows: the vertices are the
  items, and the (weighted or unweighted) distance between two items
  is a measure of how closely together the items should be
  `packed'. The notion of bin stretch is as before, and colour stretch
  is replaced by the following notion of `strong diameter stretch': for a set
  of vertices $X$, let $\mathrm{diam}(X)$ be their `strong diameter',
  i.e. $\max_{u,v \in X} d_G(u,v)$. Let $B(X)$ be the bins spanned by
  items in $X$. Then strong diameter stretch is $\max_{X \subseteq V}
  \frac{\mathrm{diam}(X)}{|B(X)|}$. What bounds can we achieve when
  using this quantity, and does it depend on eg. the expansion of $G$?

\end{itemize}

\end{document}